\documentclass[10pt, a4paper]{article}     
\usepackage{amsmath,amssymb,amsfonts} 
\usepackage{amsthm}
\usepackage[latin1]{inputenc}
\usepackage{graphics} 
\usepackage{graphicx}                
\usepackage{color}                    
\usepackage{hyperref}                 
\usepackage{lineno}

\theoremstyle{plain}
\newtheorem{thm}{Theorem}[section]
\newtheorem{lem}[thm]{Lemma}

\newtheorem{prop}[thm]{Proposition}
\theoremstyle{definition}
\newtheorem{defi}[thm]{Definition}

\theoremstyle{remark}

\newtheorem{rem}[thm]{Remark}

\renewenvironment{proof}{\textsc{Proof:}}{\qed}

\title{A probabilistic-numerical approximation for an obstacle problem arising in game theory}

\author{Christine Gr\"un \footnote{Laboratoire de Mathematiques de Brest
UMR 6205, 6 avenue Le Gorgeu CS 93837, 29238 BREST cedex 3, France; email: christine.gruen@univ-brest.fr.}
\footnote{Supported by the Marie Curie Initial Training Network (ITN) project: ``Deterministic and Stochastic Controlled Systems and
Application", FP7-PEOPLE-2007-1-1-ITN, No. 213841-2.}
\footnote{Ce travail a b\'en\'efici\'e une aide de l'Agence Nationale de la Recherche portant la r\'ef\'erence ANR-10-BLAN 0112}
}

\begin{document}
\maketitle
\begin{abstract}
We investigate a two-player zero-sum stochastic differential game in which one of the players has more information on the game than his opponent. We show how to construct numerical schemes for the value function of this game, which is given by the solution of a quasilinear partial differential equation with obstacle.
\end{abstract}



\section{Introduction}
In 1967 Aumann and Maschler presented their celebrated model for games with incomplete information, see \cite{AuMaS} and references therein. The game they consider consists in a set of, say $I$, standard discrete time two person zero-sum games. At the beginning one of these zero-sum games is picked at random according to a probability $p$. The information which game was picked is transmitted to Player 1 only, while Player 2 just knows $p$. It is assumed that both players observe the actions of the other one, so Player 2 might infer from the actions of his opponent which game is actually played. It turns out that it is optimal for the informed player to play with an additional randomness. Namely in a such a way, that he optimally manipulates the beliefs of the uninformed player.

The extension to two-player zero-sum stochastic differential games has recently been given by Cardalia- guet and Rainer in \cite{CaRa2}, \cite{Ca}, where the value function is characterized by the unique viscosity solution of a Hamilton Jacobi Isaacs (HJI) equation with an obstacle in the form of a convexity constraint in $p$.  The HJI equation without obstacle is the one which is also found to characterize stochastic differential games in the classical work of Fleming and Souganidis \cite{FS}. The probability $p$ appears as an additional parameter in which the value function has to be convex. 

In Cardaliaguet \cite{Carda} an approximation scheme for the value function of deterministic differential games with incomplete information is introduced. An extension of \cite{Carda} to deterministic games with information incompleteness on both sides is given in the work of  Souquiere \cite{Sou}. We consider the case where the underlying dynamic is given by a diffusion with controlled drift but uncontrolled non-degenerate volatility. In constrast to \cite{Carda} and \cite{Sou} we can work on the problem under a Girsanov transform. This transform is a well known tool to consider stochastic games with complete information in the context of backward stochastic differential equations (BSDEs) (see Hamad\`ene and Lepeltier \cite{HaLe}). An approximation of the value function of a stochastic differential game via BSDEs has been discussed in Bally \cite{Ba}.  Different to \cite{Ba} our algorithm is closely related to the work of Barles and Souganidis \cite{BaSu} who consider monotone approximation schemes for fully nonlinear second order partial differential equations. The latter was also applied in the recent work of Fahim, Touzi and Warin \cite{FTW} where fully nonlinear parabolic PDEs are treated. As in \cite{FTW} we use a kind of finite difference scheme for the HIJ backwards in time and combine it with taking the convex hull in $p$ at each time step to capture the effect of the information incompleteness. Note that this rather direct ansatz using a probabilistic PDE scheme also significantly differs from the Makov chain approximation method for stochastic differential games described in Kushner \cite{Ku}.

From the very beginning of the investigation of BSDEs initiated by Peng in \cite{P} the close relationship with optimal control problems and quasilinear PDEs has been exploited.  Consequently, also the approximation of solutions to BSDEs and to quasilinear PDEs are closely related. For a survey on BSDEs we refer to El Karoui, Peng and Quenez \cite{ElK}, while a survey on the numerical approximation of BSDEs can be found in Bouchard, Elie and Touzi \cite{BuETou}. In this sense our result can also be interpreted as approximation of the solutions to the BSDEs which appear in the BSDE representation of the value function for stochastic differential games with incomplete information in \cite{CG}.

The outline of the paper is as follows. In section 2 we describe the game and restate the results of \cite{CaRa2} and \cite{Ca} which build the basis for our investigation. In section 3 we present the approximation scheme and give some regularity proofs. Section 4 is devoted to the convergence proof.

\section{Setup}

\subsection{Formal description of the game}

Let $\mathcal{C}([t_0,T];\mathbb{R}^d)$ be the set of continuous functions from $\mathbb{R}$ to $\mathbb{R}^d$, which are constant on $(-\infty,t_0]$ and on $[T,+\infty)$. We denote by $B_s(\omega_B)=\omega_B(s)$ the coordinate mapping on $\mathcal{C}([t_0,T];\mathbb{R}^d)$ and define $\mathcal{H}=(\mathcal{H}_s)$ as the filtration generated by $s \mapsto B_s$. We denote $\Omega_t=\{\omega\in\mathcal{C}([t,T];\mathbb{R}^d)\}$ and $\mathcal{H}_{t,s}$ the $\sigma$-algebra generated by paths up to time $s$ in $\Omega_t$.  Furthermore we provide $\mathcal{C}([t_0,T];\mathbb{R}^d)$ with the Wiener measure $\mathbb{P}^0$ on $(\mathcal{H}_s)$.\\
In the following we investigate a two-player zero-sum differential game starting at a time $t\geq t_0$ with terminal time $T$.  For any fixed initial data $t\in[t_0,T], x\in\mathbb{R}^d$ the two players control a diffusion on $(\mathcal{C}([t,T];\mathbb{R}^d),(\mathcal{H}_{t,s})_{s\in[t,T]}, \mathcal{H},\mathbb{P}^0)$ given by
\begin{eqnarray}
 dX^{t,x,u,v}_s=b(s,X^{t,x,u,v}_s,u_s,v_s)ds+\sigma(s,X^{t,x,u,v}_s)dB_s\ \ \ \ 
 X^{t,x}_{t}=x.
\end{eqnarray}
where we assume that the controls of the players $u$, $v$ can only take their values in some compact subsets of some finite dimensional spaces, denoted by $U$, $V$ respectively.\\
The aim of the game is to optimize
\begin{itemize}
	\item[(i)] running costs: $(l_i)_{i\in\{1,\ldots , I\}}:[t_0,T]\times\mathbb{R}^d\times U \times V \rightarrow \mathbb{R}$
	\item[(ii)] terminal payoffs: $(g_i)_{i\in\{1,\ldots , I\}}:\mathbb{R}^d\rightarrow\mathbb{R}$,
\end{itemize}
which are chosen according to a probability $p\in \Delta(I)$ before the game starts. At the beginning of the game this information is transmitted only to Player 1. We assume that Player 1 chooses his control to minimize, Player 2 chooses his control to maximize the expected payoff.  Furthermore we assume both players observe their opponents control. So Player 2, knowing only the probability $p_{i}$ for scenario $i\in\{1,\ldots , I\}$ at the beginning, will try to guess the missing information from the behavior of his opponent.\\
The following will be the standing assumption throughout the paper.\\
{\bf Assumption (A)}
\begin{itemize}
	\item[(i)] $b:[t_0,T]\times\mathbb{R}^d\times U \times V \rightarrow \mathbb{R}^d$ is bounded and continuous in all its variables and Lipschitz continuous with respect to $(t,x)$ uniformly in $(u,v)$.
	\item[(ii)] For $1\leq k,l\leq d$ the function $\sigma_{k,l}:[t_0,T]\times\mathbb{R}^d \rightarrow \mathbb{R}$ is bounded and Lipschitz continuous with respect to $(t,x)$. For any $(t,x)\in[0,T]\times\mathbb{R}^d$ the matrix $\sigma^*(t,x)$ is non-singular and $(\sigma^*)^{-1}(t,x)$ is bounded and Lipschitz continuous with respect to $(t,x)$.
	\item[(iii)] $(l_i)_{i\in I}:[t_0,T]\times\mathbb{R}^d\times U \times V \rightarrow \mathbb{R}$ is bounded and continuous in all its variables and Lipschitz continuous with respect to $(t,x)$ uniformly in $(u,v)$. $(g_i)_{i\in I}:\mathbb{R}^d \rightarrow \mathbb{R}$ is bounded and uniformly Lipschitz continuous.
	\item[(iv)] Isaacs condition: for all $(t,x,\xi,p)\in[t_0,T]\times\mathbb{R}^d\times\mathbb{R}^d\times\Delta(I)$ 
	\begin{equation}
	\begin{array}{rcl}
		&&\inf_{u\in U}\sup_{v\in V} \left\{\langle b(t,x,u,v),\xi\rangle+\sum_{i=1}^{I}p_il_i(t,x,u,v)\right\}\\
		\ \\
		&&\ \ \ \ =\sup_{v\in V} \inf_{u\in U}\left\{\langle b(t,x,u,v),\xi\rangle+\sum_{i=1}^{I}p_il_i(t,x,u,v)\right\}=:H(t,x,\xi,p).
	\end{array}
	\end{equation}
\end{itemize}

By assumption (A) the Hamiltonian $H$ is  Lipschitz continuous in $(\xi,p)$ uniformly in $(t,x)$ and Lipschitz continuous in $(t,x)$ with Lipschitz constant $c(1+|\xi|)$, i.e. it holds for all $t,t'\in[0,T]$, $x,x'\in\mathbb{R}^d$, $\xi,\xi'\in\mathbb{R}^d$, $p,p'\in\Delta(I)$
\begin{eqnarray}
|H(t,x,\xi,p)|\leq c (1+|\xi|)
\end{eqnarray}
and
\begin{eqnarray}
|H(t,x,\xi,p)-H(t',x',\xi',p')|\leq c (1+|\xi|)(|x-x'|+|t-t'|)+c|\xi-\xi'|+ c |p-p'|.
\end{eqnarray}

\subsection{Strategies and value function}

We now give the necessary definitions and the results of \cite{Ca} and \cite{CaRa2} on which we will base our investigation.

\begin{defi}
For any $t\in[t_0,T[$ an admissible control $u=(u_s)_{s\in[t,T]}$ for Player 1 is a progressively measurable process with respect to the filtration $(\mathcal{H}_{t,s})_{s\in[t,T]}$ with values in $U$.
The set of admissible controls for Player 1 is denoted by $\mathcal{U}(t)$.\\
The definition for admissible controls $v=(v_s)_{s\in[t,T]}$ for Player 2 is similar. The set of admissible controls for Player 2 is denoted by $\mathcal{V}(t)$.
\end{defi}


\begin{defi}
A strategy for Player 1 at time $t\in[t_0,T[$ is a map $\alpha:[t,T]\times\mathcal{C}([t,T];\mathbb{R}^d)\times L^0([t,T];V)\rightarrow U$ which is nonanticipative with delay, i.e. there is $\delta>0$ such that for all $s\in[t,T]$ for any $f,f'\in\mathcal{C}([t,T];\mathbb{R}^d)$ and $g,g'\in L^0([t,T];V)$ it holds: $f=f'$  and $g=g'$ a.e. on $[t,s]$ $\Rightarrow$ $\alpha(\cdot,f,g)=\alpha(\cdot,f',g')$ a.e. on  $[t,s+\delta]$. The set of strategies for Player 1 is denoted by $\mathcal{A}(t)$.\\
The definition of strategies $\beta: [t,T]\times\mathcal{C}([t,T];\mathbb{R}^d)\times L^0([t,T];U)\rightarrow V$ for Player 2 is similar. The set of strategies for Player 2 is denoted by $\mathcal{B}(t)$.
\end{defi}

With Definition 2.2. it is possible to prove via a fixed point argument the following Lemma, which is a slight modification of Lemma 5.1. in \cite{CaRa2}.
\begin{lem}
To each pair of strategies $(\alpha,\beta)\in\mathcal{A}(t)\times\mathcal{B}(t)$ one can associate a unique couple of admissible controls $(u,v)\in\mathcal{U}(t)\times\mathcal{V}(t)$, such that for all $\omega\in\mathcal{C}([t,T];\mathbb{R}^d)$
\[\alpha(s,\omega,v(\omega))=u_s(\omega)\ \ \ \ \textnormal{ \textit{and}}\ \ \ \ \beta(s,\omega,u(\omega))=v_s(\omega)\ .\]
\end{lem}

A characteristic feature of games with incomplete or asymmetric information is that the players have to find a balance between acting optimally according to their information and hiding it. To this end it turns out that he will give his behavior a certain additional randomness. This effect is captured in the following definition. 

\begin{defi}
A random strategy for Player 1 at time $t\in[t_0,T[$ is a a pair $((\Omega_\alpha,\mathcal{G}_\alpha,\mathbb{P}_\alpha),\alpha)$, where $(\Omega_\alpha,\mathcal{G}_\alpha,\mathbb{P}_\alpha)$ is a probability space in $\mathcal{I}$ and $\alpha: [t,T]\times\Omega_\alpha\times\mathcal{C}([t,T];\mathbb{R}^d)\times L^0([t,T]; V)\rightarrow U$ satisfies
\begin{itemize}
\item[(i)] 
 	$\alpha$ is a measurable function, where $\Omega_\alpha$ is equipped with the $\sigma$-field $\mathcal{G}_\alpha$,
\item[(ii)] 
	there exists $\delta>0$ such that for all $s\in[t,T]$ and for any $f,f'\in\mathcal{C}([t,T];\mathbb{R}^d)$ and $g,g'\in L^0([t,T];V))$ it holds: 
        \center{$f=f'$  and $g=g'$ a.e. on $[t,s]$ $\Rightarrow$ $\alpha(\cdot,f,g)=\alpha(\cdot,f',g')$ a.e. on  $[t,s+\delta]$ for any $\omega\in\Omega_\alpha$.}
\end{itemize}
The set of random strategies for Player 1 is denoted by $\mathcal{A}^r(t)$.\\
The definition of random strategies $((\Omega_\beta,\mathcal{G}_\beta,\mathbb{P}_\beta),\beta)$, where $\beta: [t,T]\times \Omega_\beta \times\mathcal{C}([t,T];\mathbb{R}^d)\times L^0([t,T];U)\rightarrow V$ for Player 2 is similar. The set of random strategies for Player 2 is denoted by $\mathcal{B}^r(t)$.
\end{defi}

\begin{rem} Again one can associate to each couple of random strategies $(\alpha,\beta)\in\mathcal{A}^r(t)\times\mathcal{B}^r(t)$ for any $(\omega_\alpha,\omega_\beta)\in\Omega_\alpha\times\Omega_\beta$ a unique couple of admissible strategies $(u^{\omega_\alpha,\omega_\beta},v^{\omega_\alpha,\omega_\beta})\in\mathcal{U}(t)\times\mathcal{V}(t)$, such that for all $\omega\in\mathcal{C}([t,T];\mathbb{R}^d)$, $s\in[t,T]$
\begin{equation*}
\alpha(s,\omega_\alpha,\omega,v^{\omega_\alpha,\omega_\beta}(\omega))=u^{\omega_\alpha,\omega_\beta}_s(\omega)\ \ \ \  \textnormal{ and }\ \ \ \  \beta(s,\omega_\beta,\omega,u^{\omega_\alpha,\omega_\beta}(\omega))=v^{\omega_\alpha,\omega_\beta}_s(\omega)\ .
\end{equation*}
Furthermore $(\omega_\alpha,\omega_\beta)\rightarrow (u^{\omega_\alpha,\omega_\beta},v^{\omega_\alpha,\omega_\beta})$ is a measurable map, from $\Omega_\alpha\times\Omega_\beta$ equipped with the $\sigma$-field $\mathcal{G}_\alpha\otimes\mathcal{G}_\beta$  to $\mathcal{V}(t)\times\mathcal{U}(t)$ equipped with the  Borel $\sigma$-field associated to the $L^1$-distance.
\end{rem}

For any $(t,x,p)\in[t_0,T[\times\mathbb{R}^d\times\Delta(I)$, $\bar \alpha \in (\mathcal{A}^r(t))^I$, $\beta\in\mathcal{B}^r(t)$ we set
        		\begin{eqnarray}
        			J(t,x,p,\bar \alpha,\beta)=\sum_{i=1}^I p_i \ \mathbb{E}_{\bar \alpha_i,\beta}\left[\int_0^Tl_i(s,X_s^{t,x,\bar \alpha_i,\beta},  (\bar\alpha_i)_s,\beta_s)ds+g_i(X_T^{t,x,	\bar\alpha_i,\beta})\right],
		\end{eqnarray}
where as in Remark 2.5. we associate to ${\bar \alpha_i},\beta$ for any $(\omega_{\bar \alpha_i},\omega_\beta)\in\Omega_{\bar \alpha_i}\times\Omega_\beta$ the couple of controls $(u^{\omega_{\bar \alpha_i},\omega_\beta},v^{\omega_{\bar \alpha_i},\omega_\beta})$. The process $X^{t,x,\bar \alpha_i,\beta}$ is then defined for any $(\omega_{\bar \alpha_i},\omega_\beta)$ as solution to the SDE (1) with the associated controls. Furthermore $\mathbb{E}_{\bar \alpha_i,\beta}$ is the expectation on $\Omega_{\bar \alpha_i}\times\Omega_\beta\times\mathcal{C}([t,T];\mathbb{R}^d)$ with respect to the probability $\mathbb{P}_{\bar \alpha_i}\otimes\mathbb{P}_\beta\otimes\mathbb{P}^0$, where $\mathbb{P}^0$ denotes the Wiener measure on $\mathcal{C}([t,T];\mathbb{R}^d).$\\

Under assumption (A) the existence of the value of the game and its characterization as a viscosity solution to an obstacle problem is shown in  \cite{Ca},\cite{CaRa2}.
        	\begin{thm}
       		 For any $(t,x,p)\in[t_0,T[\times\mathbb{R}^d\times\Delta(I)$ the value of the game with incomplete information $V(t,x,p)$ is given by
			\begin{equation}
				\begin{array}{rcl}
		 			V(t,x,p) &=& \inf_{\bar \alpha \in (\mathcal{A}^r(t))^I}\sup_{\beta\in \mathcal{B}^r(t)} J(t,x,p,\bar \alpha,\beta)\\
					\ \\
					&=&\sup_{\beta\in \mathcal{B}^r(t)} \inf_{\bar \alpha \in (\mathcal{A}^r(t))^I} J(t,x,p,\bar \alpha,\beta).
				\end{array}
			\end{equation}
		Furthermore the function $V:[0,T[\times \mathbb{R}^d\times\Delta(I)\rightarrow\mathbb{R}$ is the unique viscosity solution to
\begin{equation}
	\min \left\{ \frac{\partial w} {\partial t}+\frac{1}{2}\textnormal{tr}(\sigma\sigma^*(t,x)D_x^2w)+H(t,x,D_xw,p),\lambda_{\min}\left(p,\frac{\partial ^2 w}{\partial p^2}\right)\right\}=0
\end{equation}
with terminal condition $w(T,x,p)=\sum_{i}p_ig_i(x)$, where for all $p\in\Delta(I)$, $A\in\mathcal{S}^I$
\begin{eqnarray}
\lambda_{\min}(p,A):=\min_{z\in T_{\Delta(I)(p)}\setminus\{0\}} \frac{\langle Az,z\rangle}{|z|^2}.
\end{eqnarray}
and  $T_{\Delta(I)(p)}$ denotes the tangent cone to $\Delta(I)$ at $p$, i.e. $T_{\Delta(I)(p)}=\overline{\cup_{\lambda>0}(\Delta(I)-p)/\lambda}$ .
		
		        \end{thm}

\begin{rem}
Unlike the standard definition of viscosity solutions (see e.g. \cite{CIL}) the subsolution property to (7) is required only on the interior of $\Delta(I)$ while the supersolution property to (7) is required on the whole domain $\Delta(I)$ (see \cite{Ca} and \cite{CaRa2}). This is due to the fact that we actually consider viscosity solutions with a state constraint, namely $p\in\Delta(I)\subsetneq\mathbb{R}^I$. For more details we refer to \cite{CaDoL}.
\end{rem}

\section{Approximation of the value function}

\subsection{Numerical scheme}

Our approximation scheme of the value function basically amounts to approximate the solution of the obstacle problem (7). In order to do so it is convenient to consider the real dynamics of the game (1) under a Girsanov transform. This technique  - first applied to stochastic differential games by \cite{HaLe} - enables us to decouple the forward dynamics (1) from the controls of the players. As in \cite{Ba} where this transformation is applied in the context of numerical approximation for stochastic differential games via BSDE we will use the following approximation for the forward dynamics
.\\

For $L\in\mathbb{N}$ we define a partition of $[t_0,T]$ with stepsize $\tau=\frac{T}{L}$ by $\Pi^\tau=\{t_0,t_1,\ldots ,t_L=T\}$. Then for all $k=0,\ldots ,L$, $x\in\mathbb{R}^d$, $p\in\Delta(I)$ let $(X^{t_k,x}_s)_{s\in[t_k,T]}$ denote the diffusion 
\begin{eqnarray}
 X^{t_k,x}_s=x+\int_{t_k}^s \sigma(r,X_r^{t,x})dB_r.
\end{eqnarray}
Furthermore we define the discrete process $(\bar X^{k,x}_{n})_{n=k,\ldots , L}$ as the standard Euler scheme approximation for (9) on $\Pi^\tau$
\begin{eqnarray}
\bar X^{k,x}_{n}=x+\sum_{j=k}^{n-1} \sigma(t_{j},\bar X^{k,x}_{j})\Delta B^{j},
\end{eqnarray}
where $\Delta B^{j}= B_{t_{j+1}}- B_{t_{j}}$.\\
We will approximate the value function (6) backwards in time. To do so we set for all $x\in\mathbb{R}^d$, $p\in\Delta(I)$
\begin{equation}
V^\tau(t_L,x,p)=\langle p, g(x)\rangle
\end{equation}
 and we define recursively for $k=L-1,\ldots,0$
 \begin{equation}
\begin{array}{rcl}
V^\tau(t_{k-1},x,p)&=&\textnormal{Vex}_p\left( \mathbb{E}\left[V^\tau(t_{k},\bar X^{{k-1},x}_{k},p)\right]+\tau H(t_{k-1},x,\bar z_{k-1}(x,p),p)\right),
\end{array}
\end{equation}
where $\bar z_{k-1}(x,p)$ is given by 
\begin{eqnarray}
\bar z_{k-1}(x,p)=\frac{1}{\tau}\mathbb{E}\left[V^\tau(t_{k},\bar X^{{k-1},x}_{k},p) (\sigma^*)^{-1}(t_{k-1},x)\Delta B^{k-1}\right]
\end{eqnarray}
and $\textnormal{Vex}_p$ denotes the convex hull, i.e. the largest function that is convex in the variable $p$ and does not exceed the given function.

\subsection{Some regularity properties}

\subsubsection{Monotonicity}

First we show that our scheme fulfills a monotonicity condition which corresponds to the one in \cite{BaSu} (2.2). It is well known that this criteria is crucial for the convergence of general finite difference schemes.\\

\begin{lem}
Let $\phi:\mathbb{R}^d\rightarrow\mathbb{R}$ be a uniformly Lipschitz continuous function with Lipschitz constant $M$. Then there exists for all $x,x'\in \mathbb{R}^d$ a $\theta\in\mathbb{R}^d$ with $|\theta|\leq M$
\begin{equation*}
\phi(x)-\phi(x')=\langle \theta, x-x'\rangle
\end{equation*}
\end{lem}
\begin{proof}
For $\phi\in C^1$ the result follows from partial integration with $\theta=\int_0^1D_x\phi(x+r(x'-x))dr$. For the case of general Lipschitz continuous function $\phi$ one chooses a sequence of $\mathcal{C}^1$ functions $(\phi^\epsilon)_{\epsilon>0}$ which converges uniformly to $\phi$. Since $\phi$ is uniformly Lipschitz continuous, we may assume that the absolute value of $D_x\phi^\epsilon$ and hence the corresponding $\theta^\epsilon$ are uniformly bounded by the constant $M$. Consequently, possibly passing though a subsequence, there exists a $\theta\in\mathbb{R}^d$ with $|\theta|\leq M$ such that the lemma holds.
\end{proof}

With the help of Lemma 3.1 we now establish:

\begin{lem}
Let $k\in\{0,\ldots,L-1\}$ and $\phi,\psi: \mathbb{R}^d\rightarrow\mathbb{R}$ be two Lipschitz continuous functions. Then for any $x\in\mathbb{R}$, $p\in\Delta(I)$
\begin{eqnarray*}
	&&\mathbb{E}\left[\phi(\bar X^{k,x}_{{k+1}})\right]+\tau H(t_k,x,\frac{1}{\tau} \mathbb{E}\left[ \phi(\bar X^{k, x}_{{k+1}})  (\sigma^*)^{-1}(t_k,x) \Delta B^{k}\right],p)\\
	 &&\ \ \ \geq \mathbb{E}\left[\psi(\bar X^{k,x}_{{k+1}})\right]+\tau H(t_k,x,\frac{1}{\tau} \mathbb{E}\left[ \psi(\bar X^{k,x}_{{k+1}})  (\sigma^*)^{-1}(t_k,x) \Delta B^{k}\right],p)-\tau \mathcal{O}(\tau),
\end{eqnarray*}
where $\mathcal{O}(\tau)$ is independent of $p$.
\end{lem}

\begin{proof}
By (4) $H$ is uniformly Lipschitz continuous in $\xi$. So by Lemma 3.1. there exists a $\theta\in\mathbb{R}^d$ with $|\theta|\leq M$, where $M$ denotes the Lipschitz constant of $H$, such that
\begin{eqnarray*}
	&&\mathbb{E}\left[(\phi-\psi)(\bar X^{k,x}_{{k+1}})\right] + \tau \bigg(H(t_k,x,\frac{1}{\tau} \mathbb{E}\left[ \phi(\bar X^{k,x}_{{k+1}})  (\sigma^*)^{-1}( t_k,x) \Delta B^{k}\right],p)\\
	&&\ \ \ \ \ - H(t_k,x,\frac{1}{\tau} \mathbb{E}\left[ \psi(\bar X^{k,x}_{{k+1}})  (\sigma^*)^{-1}( t_k,x) \Delta B^{k}\right],p))\bigg)\\
	&&=\mathbb{E}\left[(\phi-\psi)(\bar X^{k,x}_{{k+1}})\right]+ \left \langle \tau \theta,\left(\frac{1}{\tau} \mathbb{E}\left[ \phi(\bar X^{k,x}_{{k+1}})  (\sigma^*)^{-1}( t_k,x) \Delta B^{k}\right]-\frac{1}{\tau} \mathbb{E}\left[ \psi(\bar X^{k,x}_{{k+1}})  (\sigma^*)^{-1}( t_k,x) \Delta B^{k}\right]\right)\right\rangle\\
	&&=\mathbb{E}\left[(\phi-\psi)(\bar X^{k,x}_{{k+1}})\right]+\left\langle \theta , \mathbb{E}\left[(\phi-\psi)(\bar X^{k,x}_{{k+1}})(\sigma^*)^{-1}(t_k,x)\Delta B^k\right]\right\rangle\\
	&&=\mathbb{E}\left[(\phi-\psi)(\bar X^{k,x}_{{k+1}})\left(1+\langle \theta ,(\sigma^*)^{-1}(t_k,x) \Delta B^k\rangle \right)\right].
\end{eqnarray*}
Since $0\leq \phi(x)-\psi(x) \leq c$ for any $x\in\mathbb{R}$, we have
\begin{eqnarray*}
&&\mathbb{E}\left[(\phi-\psi)(\bar X^{k,x}_{{k+1}})\left(1+\langle \theta, (\sigma^*)^{-1}(t_k,x) \Delta B^k\rangle\right) \right]\\
&\geq&\mathbb{E}\left[(\phi-\psi)(\bar X^{k,x}_{t_{k+1}}) 1_{|\Delta B^k|\geq\|\theta\sigma^{-1}\|_\infty^{-1}}\langle \theta, (\sigma^*)^{-1}(t_k,x) \Delta B^k\rangle \right]\\
&\geq& - C \mathbb{E}\left[1_{|\Delta B^k|\geq {\frac{1}{C}}} |\Delta B^k|\right] 
\end{eqnarray*}
with  $C:=\|M\sigma^{-1}\|_\infty$ independent of $(t_k,x,p)$ and $\tau$. Furthermore we can explicitely calculate
\begin{eqnarray*}
 \mathbb{E}\left[1_{|\Delta B^k|\geq  \frac{1}{C}} |\Delta B^k|\right]= \ \frac{1}{(2\pi)^\frac{d}{2}(\tau)^\frac{1}{2}}\int_{|x|\geq\frac{1}{C}}^\infty |x| e^{-\frac{x^2}{2\tau}}dx
=\frac{1}{2^{\frac{d}{2}-1}\Gamma(\frac{d}{2})}\tau^\frac{1}{2} e^{-\frac{1}{2 C^2 \tau}},
\end{eqnarray*}
where $\Gamma$ denotes the gamma function.
\end{proof}

\subsubsection{Lipschitz continuity in $x$}
To show that the Lipschitz continuity in $x$ is preserved under the scheme, we establish the following Lemma.

\begin{lem}
Let $k\in\{0,\ldots,L-1\}$ and $\phi:\mathbb{R}^d\rightarrow\mathbb{R}$ be a uniformly Lipschitz continuous function with Lipschitz constant $M$. Then for any $k\in\{0,\ldots,L-1\}$, $x,x'\in\mathbb{R}$,  $p\in\Delta(I)$
\begin{eqnarray*}
	&&\bigg|\mathbb{E}\left[\phi(\bar X^{k,x}_{{k+1}})\right]+\tau H(t_k,x,\frac{1}{\tau} \mathbb{E}\left[ \phi(\bar X^{k, x}_{{k+1}})  (\sigma^*)^{-1}(t_k,x) \Delta B^{k}\right],p)\\
	&&\ \ \ \ -\mathbb{E}\left[\phi(\bar X^{k,x'}_{{k+1}})\right]-\tau H(t_k,x',\frac{1}{\tau} \mathbb{E}\left[\phi(\bar X^{k, x'}_{{k+1}})  (\sigma^*)^{-1}(t_k,x') \Delta B^{k}\right],p)\bigg|\\
	&&\ \ \ \ \ \ \ \ \ \ \ \ \ \leq C^{M,\tau} |x-x'|,
\end{eqnarray*}
where $C^{M,\tau}=M(1+c\tau)+c \tau$ with $c$ independent of $p$.
\end{lem}

\begin{proof}
We fix $k\in\{0,\ldots,L-1\}$, $x,x'\in\mathbb{R}$,  $p\in\Delta(I)$ and write
\begin{equation}
\begin{array}{rcl}
	&&\mathbb{E}\left[\phi(\bar X^{k,x}_{{k+1}})-\phi(\bar X^{k,x'}_{{k+1}})\right] + \tau \bigg(H(t_k,x,\frac{1}{\tau} \mathbb{E}\left[ \phi(\bar X^{k,x}_{{k+1}})  (\sigma^*)^{-1}( t_k,x) \Delta B^{k}\right],p)\\
	&&\ \ \ \ \ - H(t_k,x',\frac{1}{\tau} \mathbb{E}\left[ \phi(\bar X^{k,x'}_{{k+1}})  (\sigma^*)^{-1}( t_k,x') \Delta B^{k}\right],p)\bigg)\\
	&&=\mathbb{E}\left[\phi(\bar X^{k,x}_{{k+1}})-\phi(\bar X^{k,x'}_{{k+1}})\right]\\
	&&\ \ \  + \tau \bigg(H(t_k,x,\frac{1}{\tau} \mathbb{E}\left[ \phi(\bar X^{k,x}_{{k+1}})  (\sigma^*)^{-1}( t_k,x) \Delta B^{k}\right],p)\\
	&&\ \ \  \ \ \  \ \ \  \ \ \  \ \ \  \ \ \  \ \ \  - H(t_k,x,\frac{1}{\tau} \mathbb{E}\left[ \phi(\bar X^{k,x'}_{{k+1}})  (\sigma^*)^{-1}( t_k,x') \Delta B^{k}\right],p)\bigg)\\
	&&\ \ \  + \tau \bigg(H(t_k,x,\frac{1}{\tau} \mathbb{E}\left[ \phi(\bar X^{k,x'}_{{k+1}})  (\sigma^*)^{-1}( t_k,x') \Delta B^{k}\right],p)\\
	&&\ \ \  \ \ \  \ \ \  \ \ \  \ \ \  \ \ \  \ \ \   - H(t_k,x',\frac{1}{\tau} \mathbb{E}\left[ \phi(\bar X^{k,x'}_{{k+1}})  (\sigma^*)^{-1}( t_k,x') \Delta B^{k}\right],p)\bigg).
\end{array}
\end{equation}
Assume that $\phi\in \mathcal{C}^1$ with $|D_x\phi|\leq M$. First we consider the last term of (14). We have for $\Theta^1:=\int_0^1 D_x\phi(x'+r \sigma(t_{k},x')\Delta B^{k}) dr$ that $|\Theta^1|\leq M$ and
\begin{eqnarray*}
\left|\frac{1}{\tau} \mathbb{E}\left[ \phi(\bar X^{k,x'}_{{k+1}})  (\sigma^*)^{-1}( t_k,x') \Delta B^{k}\right]\right|
& =& \frac{1}{\tau}\left|\mathbb{E}\left[\phi(x') (\sigma^*)^{-1}(t_{k},x')\Delta B^{k}+\Theta^1 |\Delta B^{k}|^2\right]\right|\\
& \leq& M.
\end{eqnarray*}
Since by (4) the Hamiltonian $H$ is uniformly Lipschitz continuous in $x$ with Lipschitz constant $c(1+|\xi|)$ it holds
\begin{eqnarray*}
&&\tau \bigg(H(t_k,x,\frac{1}{\tau} \mathbb{E}\left[ \phi(\bar X^{k,x'}_{{k+1}})  (\sigma^*)^{-1}( t_k,x') \Delta B^{k}\right],p)- H(t_k,x',\frac{1}{\tau} \mathbb{E}\left[ \phi(\bar X^{k,x'}_{{k+1}})  (\sigma^*)^{-1}( t_k,x') \Delta B^{k}\right],p)\bigg)\\
	&&\leq \tau c(1+M) |x-x'|.
\end{eqnarray*}

For the remaining terms in (14) we note that by (4) the Hamiltonian $H$ is uniformly Lipschitz continuous. So there exists as in Lemma 3.1. a $\theta^1\in\mathbb{R}^d$ with $|\theta^1|\leq c$, such that
\begin{eqnarray}
&&\mathbb{E}\left[\phi(\bar X^{k,x}_{{k+1}})-\phi(\bar X^{k,x'}_{{k+1}})\right]\nonumber\\
	&&\ \ \  + \tau \bigg(H(t_k,x,\frac{1}{\tau} \mathbb{E}\left[ \phi(\bar X^{k,x}_{{k+1}})  (\sigma^*)^{-1}( t_k,x) \Delta B^{k}\right],p)- H(t_k,x,\frac{1}{\tau} \mathbb{E}\left[ \phi(\bar X^{k,x'}_{{k+1}})  (\sigma^*)^{-1}( t_k,x') \Delta B^{k}\right],p)\bigg)\nonumber\\
	&&=\mathbb{E}\left[\phi(\bar X^{k,x}_{{k+1}})-\phi(\bar X^{k,x'}_{{k+1}})\right] \nonumber \\
	&&\ \ \ \ \ \ \ \ \ + \langle \theta^1, 
	 \left(\mathbb{E}\left[ \phi(\bar X^{k,x}_{{k+1}})  (\sigma^*)^{-1}( t_k,x) \Delta B^{k}\right]-\mathbb{E}\left[ \phi(\bar X^{k,x'}_{{k+1}})  (\sigma^*)^{-1}( t_k,x') \Delta B^{k}\right]\right)\rangle\nonumber\\
	&&=\mathbb{E}\left[(\phi(\bar X^{k,x}_{{k+1}})-\phi(\bar X^{k,x'}_{{k+1}}) )(1+ \langle \theta^1, (\sigma^*)^{-1}( t_k,x) \Delta B^{k}\rangle)\right]\\
	&&\ \ \ \ \ \ \ \ + \mathbb{E}\left[\langle \theta^1, \phi(\bar X^{k,x'}_{{k+1}})  ((\sigma^*)^{-1}( t_k,x)-(\sigma^*)^{-1}( t_k,x')) \Delta B^{k}\rangle\right]\nonumber.
\end{eqnarray}
For the first term of (15) we have with $\Theta^2:=\int_0^1D_x\phi(\bar X^{k,x}_{{k+1}}+r(\bar X^{k,x'}_{{k+1}}-\bar X^{k,x}_{{k+1}}))dr$ 
\begin{eqnarray*}
&&\mathbb{E}\left[(\phi(\bar X^{k,x}_{{k+1}})-\phi(\bar X^{k,x'}_{{k+1}}) )(1+ \langle \theta^1,  (\sigma^*)^{-1}( t_k,x) \Delta B^{k}\rangle)\right]\\
	&&=\mathbb{E}\left[\left\langle \Theta^2, \bar X^{k,x}_{{k+1}}-\bar X^{k,x'}_{{k+1}} \right\rangle(1+ \langle \theta^1,(\sigma^*)^{-1}( t_k,x) \Delta B^{k}\rangle)\right]\\
	&&\leq\mathbb{E}\left[\left\langle \Theta^2,(1+\langle \theta^1,(\sigma^*)^{-1}( t_k,x) \Delta B^{k}\rangle)(x-x')+ (\sigma(t_k,x)-\sigma(t_k,x')) \Delta B^k\right\rangle\right]\\
	&&\ \ \ \ \ \ \ \ \ \ \ \ \ +c\tau|x-x'|.
\end{eqnarray*}
We finally use Cauchy-Schwartz (note that in the expansion of the square the $\Delta B^k$ parts vanish when taking expectation), $|\Theta^2|\leq M$ and the Lipschitz contiunity of $\sigma$ to get
\begin{eqnarray*}
&&\mathbb{E}\left[\left\langle \Theta^2, (1+\langle \theta^1,(\sigma^*)^{-1}( t_k,x) \Delta B^{k}\rangle)(x-x')+ (\sigma(t_k,x)-\sigma(t_k,x')) \Delta B^k\right\rangle\right]\\
&&\leq M \mathbb{E}\left[\left((1+\langle \theta^1,(\sigma^*)^{-1}( t_k,x) \Delta B^{k}\rangle)(x-x')+ (\sigma(t_k,x)-\sigma(t_k,x')) \Delta B^k\right)^2\right]^\frac{1}{2}\\
	&&\leq  M|x-x'| \left(\mathbb{E}\left[1+c|\Delta B^k|^2\right]\right)^\frac{1}{2}= M|x-x'| (1+c\tau)^\frac{1}{2}\leq M|x-x'| (1+\frac{c}{2}\tau).
\end{eqnarray*}
For the second term of (15) we use the uniform Lipschitz continuity of $(\sigma^*)^{-1}$ (by assumption (A)) to have with the $\mathbb{R}^d$-valued random variable $\Theta^3:=\int_0^1D_x\phi(\bar X^{k,x'}_{{k+1}}+r(\bar X^{k,x'}_{{k+1}}-x'))dr$
\begin{eqnarray*}
&&\mathbb{E}\left[\langle \theta^1, \phi(\bar X^{k,x'}_{{k+1}})  ((\sigma^*)^{-1}( t_k,x)-(\sigma^*)^{-1}( t_k,x')) \Delta B^{k}\rangle\right]\\
	&&=\mathbb{E}\left[ \langle \theta^1, \phi(\bar X^{k,x'}_{{k+1}})  ((\sigma^*)^{-1}(t_k,x)-(\sigma^*)^{-1}(t_k,x')) \Delta B^{k}\rangle\right]\\
	&&=\mathbb{E}\left[\langle \theta^1, (\phi(x')+\langle \Theta^3, \sigma(t,x') \Delta B^{k}\rangle)((\sigma^*)^{-1}(t_k,x)-(\sigma^*)^{-1}(t_k,x')) \Delta B^{k}\rangle \right]\\
	&&\leq c M \tau |x-x'|.
\end{eqnarray*}
The case of Lipschitz continuous $\phi$ follows by approximation with a sequence of $\mathcal{C}^1$ functions $(\phi^\epsilon)_{\epsilon>0}$ which converges uniformly to $\phi$. Since $\phi$ is uniformly Lipschitz continuous with constant $M$, we may assume that $|D_x\phi^\epsilon|\leq M$ for all $\epsilon>0$.
\end{proof}

With the previous Lemma it is easy to show the Lipschitz continuity of $V^\tau(t_{\cdot},x,p)$ in $x$.

\begin{prop}
$V^\tau(t_{\cdot},x,p)$ is uniformly Lipschitz continuous in $x$ with a Lipschitz constant that depends only on the constants of assumption (A).
\end{prop}

\begin{proof}
We will show Proposition 3.5. by induction. With (A) we have that $V^\tau(t_L,x,p)$ is Lipschitz continuous in $x$ with a constant $M_L$ that depends only on the constants of assumption (A). Let $M_k$ be the Lipschitz constant for $V^\tau(t_k,\cdot,p)$ then by (12) and Lemma 3.3. and since Vex is monotonic, we have
\begin{eqnarray*}
|V^\tau(t_{k-1},x,p)-V^\tau(t_{k-1},x',p)|\leq M_k((1+c\tau)^\frac{1}{2}+c\tau)+c \tau)|x-x'|.
\end{eqnarray*}
Hence $M_{k-1}:=M_k(1+c\tau)+c \tau$ is a Lipschitz constant for $V^\tau(t_{k-1},\cdot,p)$ and $M:=M_L C e^{{CT}}$  for a $C$ independent of $\tau,x,p$ is a constant dominating the recursively defined Lipschitz constants $(M_k)_{k=0,\ldots,L}$.
\end{proof}
\ \\
\ \\
With the uniform Lipschitz continuity of $V^\tau$ in $x$ it follows that the value function is uniformly bounded.

\begin{prop}
$V^\tau(t_{\cdot},x,p)$ is uniformly bounded by a constant only depending on the constants of assumption (A).
\end{prop}

\begin{proof}
Fix $k\in\{0,L-1\}$, $x\in\mathbb{R}^d$, $p\in\Delta(I)$. Assume first that $V^\tau$ is at $t_{k+1}$ continuously differentiable in the second variable  with $|D_xV^\tau|\leq M$. Then with $\Theta:=\int_0^1 D_xV^\tau(t_{k+1},x+r \sigma(t_{k},x)\Delta B^{k},p) dr$
\begin{equation}
\begin{array}{rcl}
|\bar z_k(x,p)|
&=&  \frac{1}{\tau}\left| \mathbb{E}\left[V^\tau(t_{k+1},x+\sigma(t_k,x)\Delta B^k,p) (\sigma^*)^{-1}(t_{k},x)\Delta B^{k}\right]\right|\\
\ \\
& =& \frac{1}{\tau}\left|\mathbb{E}\left[V^\tau(t_{k+1},x,p) (\sigma^*)^{-1}(t_{k},x)\Delta B^{k}+\Theta |\Delta B^{k}|^2\right]\right|\\
\ \\
& \leq& M.
\end{array}
\end{equation}
Since $V^\tau$ is by Lemma 3.3. uniformly Lipschitz continuous in $x$ one has (16) in the general case again by regularization.\\
By (A) $V^\tau(t_L,x,p)$ is bounded by  a constant $M_L$ that depends only on the constants of assumption (A). Let $M_k$ be a bound for $|V^\tau(t_k,\cdot,p)|$ then by (3) the definition (12)  and (16) we have
\begin{eqnarray*}
\mathbb{E}\left[V^\tau(t_{k},\bar X^{{k-1},x}_{k},p)\right]+\tau H(t_{k-1},x,\bar z_{k-1}(x,p),p)\leq M_k +   c \tau(1+M) 
\end{eqnarray*}
and $M_L + c T (1+M)$ is a constant dominating the recursively defined constants $(M_k)_{k=0,\ldots,L}$.
\end{proof}

\subsubsection{Lipschitz continuity in $p$}

The Lipschitz continuity of $V^\tau(t_{\cdot},x,p)$ in $p$ can be shown with similar methods.

\begin{lem}
Let $k\in\{0,\ldots,L-1\}$ and $\phi:\mathbb{R}^d\times\Delta(I)\rightarrow\mathbb{R}$ be a uniformly Lipschitz continuous function with Lipschitz constant $M$. Then for any $k\in\{0,\ldots,L-1\}$, $x\in\mathbb{R}^d$, $p,p'\in\Delta(I)$
\begin{eqnarray*}
	&&\bigg|\mathbb{E}\left[\phi(\bar X^{k,x}_{{k+1}},p)\right]+\tau H(t_k,x,\frac{1}{\tau} \mathbb{E}\left[ \phi(\bar X^{k, x}_{{k+1}},p)  (\sigma^*)^{-1}(t_k,x) \Delta B^{k}\right],p)\\
	&&\ \ \ \ -\mathbb{E}\left[\phi(\bar X^{k,x}_{{k+1}},p')\right]-\tau H(t_k,x,\frac{1}{\tau} \mathbb{E}\left[\phi(\bar X^{k, x}_{{k+1}},p')  (\sigma^*)^{-1}(t_k,x) \Delta B^{k}\right],p')\bigg|\\
	&&\ \ \ \ \ \ \ \ \ \ \ \leq \bar C^{M,\tau} |p-p'|,
\end{eqnarray*}
where $\bar C^{M,\tau}=M(1+c\tau)+c \tau$.
\end{lem}

\begin{proof}
We fix $k\in\{0,\ldots,L-1\}$, $x\in\mathbb{R}^d$,  $p,p'\in\Delta(I)$. First note that by (4) the Hamiltonian is uniformly Lipschitz in $p$. Hence
\begin{eqnarray*}
	&&\mathbb{E}\left[\phi(\bar X^{k,x}_{{k+1}},p)-\phi(\bar X^{k,x}_{{k+1}},p')\right] + \tau \bigg(H(t_k,x,\frac{1}{\tau} \mathbb{E}\left[ \phi(\bar X^{k,x}_{{k+1}},p)  (\sigma^*)^{-1}( t_k,x) \Delta B^{k}\right],p)\\
	&&\ \ \ \ \ - H(t_k,x,\frac{1}{\tau} \mathbb{E}\left[ \phi(\bar X^{k,x}_{{k+1}},p')  (\sigma^*)^{-1}( t_k,x) \Delta B^{k}\right],p')\bigg)\\
	&&\leq \mathbb{E}\left[\phi(\bar X^{k,x}_{{k+1}},p)-\phi(\bar X^{k,x}_{{k+1}},p')\right] + \tau \bigg(H(t_k,x,\frac{1}{\tau} \mathbb{E}\left[ \phi(\bar X^{k,x}_{{k+1}},p)  (\sigma^*)^{-1}( t_k,x) \Delta B^{k}\right],p)\\
	&&\ \ \ \ \ - H(t_k,x,\frac{1}{\tau} \mathbb{E}\left[ \phi(\bar X^{k,x}_{{k+1}},p')  (\sigma^*)^{-1}( t_k,x) \Delta B^{k}\right],p)\bigg)+c\tau|p-p'|.
\end{eqnarray*}

 By (4) the Hamiltonian $H$ is uniformly Lipschitz continuous in $\xi$ with a constant $c$. So by Lemma 3.1. 
\begin{eqnarray*}
	&&\mathbb{E}\left[\phi(\bar X^{k,x}_{{k+1}},p)-\phi(\bar X^{k,x}_{{k+1}},p')\right] + \tau \bigg(H(t_k,x,\frac{1}{\tau} \mathbb{E}\left[ \phi(\bar X^{k,x}_{{k+1}},p)  (\sigma^*)^{-1}( t_k,x) \Delta B^{k}\right],p)\\
	&&\ \ \ \ \ - H(t_k,x,\frac{1}{\tau} \mathbb{E}\left[ \phi(\bar X^{k,x}_{{k+1}},p')  (\sigma^*)^{-1}( t_k,x) \Delta B^{k}\right],p)\bigg)\\
	&&=\mathbb{E}\left[\phi(\bar X^{k,x}_{{k+1}},p)-\phi(\bar X^{k,x}_{{k+1}},p')\right]\\
	&&\ \ \ \ \  +\left \langle \theta,
\mathbb{E}\left[ \phi(\bar X^{k,x}_{{k+1}},p)  (\sigma^*)^{-1}( t_k,x) \Delta B^{k}\right]-\mathbb{E}\left[ \phi(\bar X^{k,x}_{{k+1}},p')  (\sigma^*)^{-1}( t_k,x) \Delta B^{k}\right]\right\rangle\\
	&&=\mathbb{E}\left[(\phi(\bar X^{k,x}_{{k+1}},p)-\phi(\bar X^{k,x}_{{k+1}},p') )(1+\langle \theta, (\sigma^*)^{-1}( t_k,x) \Delta B^{k})\rangle\right].
\end{eqnarray*}
Assume for now that  $\phi$ is differentiable in $p$ with $|D_p\phi|\leq M$. Then with $\Theta:=\int_0^1 D_p\phi(\bar X^{k,x}_{{k+1}},p+r(p-p')) dr$ we have
\begin{eqnarray*}
	&&\mathbb{E}\left[(\phi(\bar X^{k,x}_{{k+1}},p)-\phi(\bar X^{k,x}_{{k+1}},p') )(1+\langle \theta, (\sigma^*)^{-1}( t_k,x) \Delta B^{k})\rangle\right]\\
	&&=\mathbb{E}\left[\langle \Theta,(1+ \langle \theta,  (\sigma^*)^{-1}( t_k,x) \Delta B^{k}\rangle)(p-p')\rangle\right]\\
	&&\leq M |p-p'| \left(\mathbb{E}\left[1+c|\Delta B^k|^2\right]\right)^\frac{1}{2}= M|p-p'| (1+c\tau)^\frac{1}{2}\leq M|p-p'| (1+\frac{c}{2}\tau),
\end{eqnarray*}
where for the first estimate in the last line we used again Cauchy Schwartz as in the previous Lemma. The general case follows again by regularization.\\

\end{proof}

It is now easy to show the Lipschitz continuity of $V^\tau(t_{\cdot},x,p)$ in $p$ as in Proposition 3.4.

\begin{prop}
$V^\tau(t_{\cdot},x,p)$ is uniformly Lipschitz continuous in $p$ with a Lipschitz constant only depending on the constants of assumption (A).
\end{prop}

\subsubsection{H\"older continuity in $t$}

Finally we use the Lipschitz continuity of $V^\tau$ in $x$ to establish the H\"older continuity in time.


\begin{prop}
For all $L\in\mathbb{N}$, $x\in\mathbb{R}^d$, $p\in\Delta(I)$ it holds that $(t_.,x,p)\rightarrow V^\tau(t_{.},x,p)$ is H\"older continuous in $t_.$, in the sense that for all $k\in\{1,\dots,L-1\}, l\in\{1,\ldots L-k\}$, there exists a constant $c$ only depending on the constants of assumption (A), such that
\begin{equation*}
|V^\tau(t_{k+l},x,p)-V^\tau(t_k,x,p)|\leq c |t_{k+l}-t_{k}|^\frac{1}{2}.
\end{equation*}
\end{prop}

\begin{proof}
We fix $(x,p)\in\mathbb{R}^d\times\Delta(I)$. By (12), (3) and the convexity of $V^\tau$ in $p$ we have
\begin{eqnarray*}
&&|V^\tau(t_{k+l},x,p)-V^\tau(t_k,x,p)|\\
&&\ \ \ \ =\left|V^\tau(t_{k+l},x,p)-\textnormal{Vex}_p\left(\mathbb{E}\left[V^{\tau}(t_{k+1},\bar X^{k,x}_{{k+1}},p)\right]+\tau H(t_k,x,\bar z_k(x,p),p)\right)\right|\\
&&\ \ \ \ \leq \left|\mathbb{E}\left[V^\tau(t_{k+l},x,p)-V^{\tau}(t_{k+1},\bar X^{k,x}_{{k+1}},p)\right]\right|+c\tau  (1+M),
\end{eqnarray*}
where we used that by (16) $|\bar z_k(x,p)|$ is bounded uniformly in $p\in\Delta(I)$ by the Lipschitz constant of $V^{\tau}$ in $x$. Note that by definition (12)
\begin{eqnarray*}
V^{\tau}(t_{k+1},\bar X^{k,x}_{{k+1}},p)=\textnormal{Vex}_p\bigg(\mathbb{E}\left[V^{\tau}(t_{k+2},\bar X^ {{k+1},{x'}}_{{k+2}},p)\right]+\tau H(r,x',\bar z_{{k+1}}(x',p),p)\bigg)\bigg|_{x'=\bar X^ {k,x}_{{k+1} }}.
\end{eqnarray*}
Hence by (A) and the fact that $V^\tau$ is convex in $p$ we have
\begin{eqnarray*}
&&\left|V^\tau(t_{k+l},x,p)-\mathbb{E}\left[V^{\tau}(t_{k+1},\bar X^{k,x}_{{k+1}},p)\right]\right|\\
&&\ \ \ \ =\bigg| V^\tau(t_{k+l},x,p)\\
&&\ \ \ \ \ \ \ \ \ \ \ \ \ \ -\mathbb{E}\bigg[\textnormal{Vex}_p\bigg(\mathbb{E}\left[V^{\tau}(t_{k+2},\bar X^ {{k+1},{x'}}_{{k+2}},p)\right]+\tau H(r,x',\bar z_{{k+1}}(x',p),p)\bigg)\bigg|_{x'=\bar X^ {k,x}_{{k+1} }}
\bigg]\bigg|\\
&&\ \ \ \ \leq \left|V^\tau(t_{k+l},x,p)-\mathbb{E}\left[V^{\tau}(t_{k+2},\bar X^ {{k+1},\bar X^ {k,x}_{{k+1} }}_{{k+2}},p)\right]\right|+c \tau (1+M)\\
&&\ \ \ \ = \left|V^\tau(t_{k+l},x,p)-\mathbb{E}\left[V^{\tau}(t_{k+2},\bar X^ {k,x}_{{k+2}},p)\right]\right|+c\tau(1+M).
\end{eqnarray*}
Since $l\tau=|t_{k+l}-t_k|$ repeating this now $l-2$ times gives
\begin{eqnarray*}
|V^\tau(t_{k+l},x,p)-V^\tau(t_k,x,p)|&\leq& \left|V^\tau(t_{k+l},x,p)-\mathbb{E}\left[V^{\tau}(t_{k+l},\bar X^ {k,x}_{{k+l}},p)\right]\right|+c(1+M)|t_{k+l}-t_k|.
\end{eqnarray*}

Furthermore by the Lipschitz continutity of $V^\tau$ in $x$ and (A) it holds 
\begin{eqnarray*}
\left|V^\tau(t_{k+l},x,p)-\mathbb{E}\left[V^{\tau}(t_{k+l},\bar X^ {k,x}_{{k+l}},p)\right]\right|\leq  M \mathbb{E}\left[|\bar X^ {k,x}_{{k+l}}-x|\right]
\leq  c |t_{k+l}-t_k|^\frac{1}{2},
\end{eqnarray*}
hence
\begin{eqnarray*}
|V^\tau(t_{k+l},x,p)-V^\tau(t_k,x,p)|\leq M |t_{k+l}-t_k|^\frac{1}{2}+c(1+M)|t_{k+l}-t_k|.
\end{eqnarray*}
\end{proof}

\section{Convergence}

\begin{thm} Under (A) we have uniform convergence on the compact subsets of $[0,T]\times\mathbb{R}^d\times\Delta(I)$, i.e.
\begin{eqnarray}
	\lim_{\tau\downarrow0,t_k\rightarrow t,x'\rightarrow x,p'\rightarrow p} V^\tau(t_k,x',p')=V(t,x,p).
\end{eqnarray}
\end{thm}

Note that by Proposition 3.5. the family $(V^\tau,\tau>0)$ is uniformly bounded. Furthermore by Proposition 3.4., 3.7. and 3.8. the family $(V^\tau,\tau>0)$ is equicontinuous, hence by Arzela Ascoli compact for the topology of uniform convergence. Furthermore any candidate for the limit of $V^\tau$ as $\tau\downarrow0$ is as a limit of convex functions  convex in $p$.\\
Let $w:[0,T]\times\mathbb{R}^d\times\Delta(I)\rightarrow\mathbb{R}$ be a candidate for the limit. We will show that $w$ is a viscosity solution to (7). Since this property uniquely characterizes the value function $V$ the convergence follows immediately.

\subsection{One step a posteriori martingales and DPP}

By construction there exists at each time step $t_k$ for any $x\in\mathbb{R}^d$ and $p\in\Delta(I)$ a linear combination of $\pi^{k,1}(x,p),\ldots ,$ $\pi^{k,I}(x,p)\in\Delta(I)$ such that
\begin{eqnarray}
\sum_{l=1}^I \lambda_l^k(x,p)\pi^{k,l}(x,p)=p\ \ \ \ \  \sum_{l=1}^I \lambda_l^k(x,p)=1
\end{eqnarray}
and
\begin{equation}
\begin{array}{rcl}
&&V^\tau(t_{k},x,p)\\
 \ \\
&&\ \ \ \  =\sum_{l=1}^I \lambda_l^k(x,p)\left(\mathbb{E}\left[V^{\tau}(t_{k+1},\bar X^{k,x}_{{k+1}},\pi^{k,l}(x,p))\right]+\tau H(t_k,x,\bar z_k(x,\pi^{k,l}(x,p)),\pi^{k,l}(x,p))\right)
\end{array}
\end{equation}
with
\begin{eqnarray}
\bar z_k(x,\pi^{k,l}(x,p))=\frac{1}{\tau}\mathbb{E}\left[V^{\tau}(t_{k+1},\bar X^{k,x}_{{k+1}},\pi^{k,l}(x,p)) (\sigma^*)^{-1}(t_{k},x)\Delta B^{k}\right],
\end{eqnarray}
where we can choose $(x,p)\rightarrow \lambda^k(x,p)\in\Delta(I)$ and $(x,p)\rightarrow \pi^{k}(x,p)\in\Delta(I)^I$ Borel measurable.\\

\begin{defi} For all $i\in I$, $k=0,\ldots,L$, $x\in\mathbb{R}^n$ and $p\in\Delta(I)$ we define the one step feedbacks $\bold{p}^{i,x,p}_{k+1}$ as  $\Delta(I)$-valued random variables which are independent of $\sigma(B_s)_{s\in\mathbb{R}}$, such that
\begin{itemize}
\item [(i)] for $k=0,\ldots,L-1$
			\begin{itemize}
			\item [(a)] if $p_i=0$ set $\bold{p}^{i,x,p}_{k+1}={p}$
			\item [(b)]  if ${p}_i>0$:
					$\bold{p}^{i,x,p}_{k+1}\in\{\pi^{k,1}(x,p),\ldots ,\pi^{k,I}(x,p)\}$ with probability
					\begin{eqnarray*}
					&&\mathbb{P}\left[\bold{p}^{i,x,p}_{k+1}=\pi^{k,l}(x,p)|(\bold{p}^{j,x',p'}_{l})_{j\in\{1,\ldots,I\},x'\in\mathbb{R},p'\in\Delta{I},l\in\{1,\ldots,k\}}\right] =\lambda_l^k(x,p)\frac{(\pi^{k,l}(x,p))_i}{p_i}
					\end{eqnarray*}
			\end{itemize}
\item [(ii)] for $k=L$ set $\bold{p}^{i,x,p}_{L+1}=e^i$.
\end{itemize}
Furthermore we define one step a posteriori martingales $\bold{p}^{x,p}_{k+1}=\bold{p}^{\bold{i},x,p}_{k+1}$, where the index $\bold{i}$ is a random variable with law $p$, independent of $\sigma(B_s)_{s\in[0,T]}$ and $(\bold{p}^{j,x',p'}_{l})_{j\in\{1,\ldots,I\},x'\in\mathbb{R},p'\in\Delta{I},l\in\{1,\ldots,L\}}$. The martingale property is a direct consequence of the proof of the Lemma given below.
\end{defi}


The following one step dynamic programming is a direct consequence of Definition 4.2.

\begin{lem}
For all ${k}=0,\ldots, L-1$, $x\in\mathbb{R}^d$, $p\in\Delta(I)$ we have
\begin{equation}
\begin{array}{rcl}
&&V^\tau(t_{k},x,p)=\mathbb{E}\left[V^{\tau}(t_{k+1},\bar X^{k,x}_{{k+1}},\bold p^{x,p}_{k+1})+\tau H(t_k,x,\bar z_{k}(x,\bold p^{x,p}_{k+1}),\bold p^{x,p}_{k+1})\right]
\end{array}
\end{equation}
with
\begin{eqnarray}
\bar z_{k}(x,\bold p^{x,p}_{k+1})=\frac{1}{\tau}\mathbb{E}\left[V^{\tau}(t_{k+1},\bar X^{k,x}_{{k+1}},p) (\sigma^*)^{-1}(t_{k},x)\Delta B^{k}\right]\bigg|_{p=\bold p^{x,p}_{k+1}}.
\end{eqnarray}
\end{lem}

\begin{proof}
Assume $(p)_i>0$ for all $i=1,\ldots,I$. By the construction for all suitable functions $f:\Delta(I)\rightarrow\mathbb{R}$ it holds
\begin{eqnarray*}
&&\mathbb{E}[f(\bold p^{x,p}_{k+1})]=\sum_{i=1}^I\mathbb{E} \left[ 1_{\{\bold i=i\}} f(\bold p^{i,x,p}_{k+1})\right]
=\sum_{i=1}^I\mathbb{E} \left[ 1_{\{\bold i=i\}} \right]\mathbb{E} \left[f(\bold p^{i,x,p}_{k+1})\right]\\
&&\ \ \ \ =\sum_{i=1}^I p_i \sum_{l=1}^I\lambda_l^k(x,p)\frac{(\pi^{k,l}(x,p))_i}{p_i} f((\pi^{k,l}(x,p)))\\
&&\ \ \ \ =\sum_{l=1}^I \lambda_l^k(x,p)f(\pi^{k,l}(x,p))
\end{eqnarray*}
and the Lemma follows with (19).
\end{proof}

\subsection{Viscosity solution property}


\subsubsection{Viscosity subsolution property of $w$}
\begin{prop}
$w$ is a viscosity subsolution of (7) on $[0,T]\times\mathbb{R}^d\times\textnormal{Int}(\Delta(I)).$
\end{prop}

\begin{proof}
Let $\phi:[0,T]\times\mathbb{R}\times \Delta(I)\rightarrow\mathbb{R}$ be a test function such that $w-\phi$ has a strict global maximum at $(\bar{t},\bar x,\bar p)$, where $\bar p\in\textnormal{Int}(\Delta(I))$. We have to show, that
\begin{eqnarray}
	\min\bigg\{\frac{\partial \phi}{\partial t}+\frac{1}{2}\textnormal{tr}(\sigma\sigma^*(t,x)D_x^2\phi)+H(t,x,D_x\phi,p), \lambda_{\min} \left(p,\frac{\partial^2 \phi}{\partial p^2}\right)\bigg\}\geq 0
\end{eqnarray}
holds at $(\bar{t},\bar x,\bar p)$. As a limit of convex functions $w$ is convex in $p$ and we have since $\bar p\in\textnormal{Int}(\Delta (I))$
 \[\lambda_{\min}\left(\bar p,\frac{\partial ^2 \phi}{\partial p^2}(\bar{t},\bar x,\bar p)\right)\geq 0.\]\\
So it remains to show
\begin{eqnarray}
	\frac{\partial \phi}{\partial t}+\frac{1}{2}\textnormal{tr}(\sigma\sigma^*(t,x)D_x^2\phi)+H(t,x,D_x\phi,p)\geq 0.
\end{eqnarray}
Note that by standard arguments (e.g. \cite{Bar}) there exists a sequence $(\bar t_k,\bar x_k,\bar p_k)_{k\in\mathbb{N}}$ such that $\bar t_{k}=l_k \frac{T}{k}= l_k \tau \in\Pi^\tau$ converges to $\bar t$ and $(\bar x_k,\bar p_k)$ converge to $(\bar x,\bar p)$ and such that $V^\tau-\phi$ has a global maximum at $(\bar{t}_{k},\bar x_k,\bar p_k)$.\\
Define $\phi^\tau=\phi+(V^\tau(\bar{t}_{k},\bar x_k,\bar p_k)-\phi(\bar{t}_{k},\bar x_k,\bar p_k))=\phi+\Delta_\tau$. Hence for all $x\in\mathbb{R}, p\in \Delta(I)$
\begin{eqnarray*}
V^\tau(\bar t_{k}+\tau,x,p)-\phi^\tau(\bar t_{k}+\tau,x,p)\leq V^\tau(\bar t_{k},\bar x_k,\bar p_k)-\phi^\tau(\bar t_{k},\bar x_k,\bar p_k)=0.
\end{eqnarray*}
Set
\begin{eqnarray*}
\bar X_{k+1}=\bar x_k+\sigma(\bar t_k,\bar x_k) \Delta B^{l_k}
\end{eqnarray*}
and
\begin{eqnarray*}
\bar z_{k}=\frac{1}{\tau}\mathbb{E}\left[V^{\tau}(\bar t_{k}+\tau,\bar X_{k+1},\bar p_{k}) (\sigma^*)^{-1}(\bar t_k,\bar x_k)\Delta B^{l_k}\right].
\end{eqnarray*}
By the definition of $V^\tau$ (12) it holds 
\begin{eqnarray*}
0&=&\textnormal{Vex}_p\left(\mathbb{E}\left[V^{\tau}(\bar t_{k}+\tau,\bar X_{k+1}, \bar p_k)+\tau H(\bar t_{k},\bar x_k,\bar z_{k},\bar p_k)ds\right]\right)-V^{\tau}(\bar t_{k},\bar x_k,\bar p_k)\\
&\leq& \mathbb{E}\left[V^{\tau}(\bar t_{k}+\tau,\bar X_{k+1},\bar p_k)\right]+\tau H(\bar t_{k},\bar x_k,\bar z_{k},\bar p_k)-V^{\tau}(\bar t_{k},\bar x_k,\bar p_k).
\end{eqnarray*}
Hence by the monotonicity Lemma 3.2. we have for all $\tau>0$
\begin{eqnarray*}
0&\leq& \mathbb{E}\left[V^{\tau}(\bar t_{k}+\tau,\bar X_{k+1},\bar p_k)+\tau H(\bar t_{k},\bar x_k,\frac{1}{\tau}\mathbb{E}\left[V^{\tau}(\bar t_{k}+\tau,\bar X_{k+1},\bar p_k) (\sigma^*)^{-1}(\bar t_k,\bar x_k) \Delta B^{l_k}\right],\bar p_k)\right]\\
&&\ \ \ \ \ \ \ \ -V^{\tau}(\bar t_{k},\bar x_k,\bar p_k)\\
&\leq& \mathbb{E}\left[\phi^{\tau}(\bar t_{k}+\tau,\bar X_{k+1},\bar p_k)+\tau H(\bar t_{k},\bar x_k,\frac{1}{\tau}\mathbb{E}\left[\phi^{\tau}(\bar t_{k}+\tau,\bar X_{k+1},\bar p_k)  (\sigma^*)^{-1}(\bar t_k,\bar x_k)\Delta B^{l_k}\right],\bar p_k)\right]\\
&&\ \ \ \ \ \ \ \ -\phi^{\tau}(\bar t_{k},\bar x_k,\bar p_k)+\tau \mathcal{O}(\tau).
\end{eqnarray*}
By expansion of the smooth function $\phi^\tau$ we have since $\phi^\tau$ is equal to $\phi$ with the linear shift $\Delta_\tau$ the inequality (24).
\end{proof}

\subsubsection{Viscosity supersolution property of $w$}

\begin{prop}
$w$ is a viscosity supersolution of (7) on $[0,T]\times\mathbb{R}^d\times\Delta(I).$
\end{prop}

\begin{proof}
To show that $w(t,x,p)$ is a viscosity supersolution of (7) let $\phi:[0,T]\times\mathbb{R}\times\Delta(I)$ be a test function, such that $w-\phi$ has a strict global minimum at $(\bar{t},\bar x,\bar p)$ with $w(\bar{t},\bar x,\bar p)-\phi(\bar{t},\bar x,\bar p)=0$ and such that its derivatives are uniformly Lipschitz continuous in $p$.\\
We have to show, that
\begin{eqnarray}
	&&\min\bigg\{\frac{\partial \phi}{\partial t}+\frac{1}{2}\textnormal{tr}(\sigma\sigma^T(t,x)D_x^2\phi)+b(t,x)D_x\phi+H(t,x,D_x\phi,p), \lambda_{\min} \left(p,\frac{\partial^2 \phi}{\partial p^2}\right)\bigg\}\leq 0
\end{eqnarray}
holds at $(\bar{t},\bar x,\bar p)$. Observe that, if $\lambda_{\min} \left(\frac{\partial^2 \phi}{\partial p^2}\right)\leq0$ at $(\bar{t},\bar x,\bar p)$, then (25) follows immediately. So we assume now $\lambda_{\min} \left(\frac{\partial^2 \phi}{\partial p^2}\right)>0$.\\
By standard arguments (e.g. \cite{Bar}) there exists a sequence $(\bar{t}_k,\bar x_k,\bar p_k)_{k\in\mathbb{N}}$ such that $\bar t_{k}= l_k \tau \in\Pi^\tau$ converges to $\bar t$ and $(\bar x_k,\bar p_k)$ converge to $(\bar x,\bar p)$ and such that $V^\tau-\phi$ has a global minimum at $(\bar{t}_k,\bar x_k,\bar p_k)$.\\
Define $\phi^\tau=\phi+(V^\tau(\bar{t}_k,\bar x_k,\bar p_k)-\phi(\bar{t}_k,\bar x_k,\bar p_k))=\phi+\Delta_\tau$. Since the minimum is global, we have
\[V^\tau(\bar t_k+\tau,x,p)-\phi^\tau(\bar t_k+\tau,x,p)\geq V^\tau(\bar{t}_k,\bar x_k,\bar p_k)-\phi(\bar{t}_k,\bar x_k,\bar p_k)= 0.\]
Note that by the assumption $\lambda_{\min} \left(\frac{\partial^2 \phi}{\partial p^2}\right)>0$ there exists $\delta,\eta>0$ such that for all $k$ great enough we have
\begin{eqnarray}
\langle \frac{\partial^2 \phi^\tau}{\partial p^2}(t,x,p)z,z\rangle>4\delta|z|^2\ \ \ \ \ \ \forall (x,p)\in B_\eta(\bar x_k,\bar p_k),\ \  t\in[\bar t_k,\bar t_{k}+\tau],\ \ \ z\in T_{\Delta(I)(\bar p_k)}.
\end{eqnarray}
Since $\phi^\tau$ is a test function for a purely local viscosity notion, one can modify it outside a neighborhood of $(\bar t_k,\bar x_k,\bar p_k)$, such that for all $(s, x)\in [\bar t_k,T]\times\mathbb{R}^d$ the function $\phi^\tau(s,x,\cdot)$ is convex on the whole convex domain $\Delta(I)$. Thus for any $p\in\Delta(I)$ it holds
\begin{eqnarray}
V^\tau(s,x,p)\geq \phi^\tau(s,x,p)\geq \phi^\tau(s,x,\bar p_k)+\langle\frac{\partial \phi^\tau }{\partial p}(s,x,p),p-\bar p_k\rangle.
\end{eqnarray}

We proceed in several steps. 
\begin{itemize}
\item [(1)] First we show a local estimate which is stronger than (27) using (26).
\item  [(2)] In the second step we establish estimates for $\bold p_{k+1}:=\bold p_{l_k+1}^{\bar p_k,\bar x_k}$ where $\bold p_{l_k+1}^{\bar p_k,\bar x_k}$ is defined as one step martingale with initial data $(\bar t_k,\bar x_k,\bar p_k)$ as in Definition 4.2.
\item  [(3)] Then we use the estimates of the second step together with the monotonicity in Lemma 3.3. to conclude the viscosity supersolution property.
\end{itemize}

\textbf{Step 1}: 
We claim that there exist $\eta,\delta>0$, such that for all $\tau>0$ small enough (meaning $k$ great enough) it holds
\begin{eqnarray}
	V^\tau(\bar t_{k}+\tau,x,p)\geq\phi^\tau(\bar t_{k}+\tau,x,\bar p_k)+\langle\frac{\partial \phi^\tau }{\partial p}( \bar t_{k}+\tau ,x,\bar p_k),p-\bar p_k\rangle+\delta|{p}-\bar p_k|^2.
\end{eqnarray}
 for all $x\in B_\eta(\bar x_k)$, $p\in\Delta(I)$. By Taylor expansion in p
\begin{eqnarray}
 \phi^\tau(t,x,p)\geq \phi^\tau(t,x,\bar p_k)+\langle\frac{\partial \phi^\tau }{\partial p}(t,x,p),p-\bar p_k\rangle+2\delta|p-\bar p_k|^2
\end{eqnarray}
holds for $(x,p)\in B_\eta(\bar x_k,\bar p_k)$, $t\in[\bar t_k, \bar t_k+\tau]$.  Hence (28) is true locally in $p$. To establish (28) for all $p\in\Delta(I)$ we set for $p\in\Delta(I)\setminus \textnormal{Int}(B_\eta(\bar p_k))$
\[ \tilde p = \bar p_k+\frac{p-\bar p_k}{|p-\bar p_k|}\eta.\]

So by the convexity of $V^\tau$ in $p$ and (29) we have for a $\hat{ p}\in\partial {V^\tau}^-(\bar t_k,\bar x_k, \tilde p)$
\begin{eqnarray*}
V^\tau(\bar t_{k},\bar x_k,p) &\geq& V^\tau(\bar t_k,\bar x_k,\tilde p)+\langle \hat p,p-\tilde p \rangle\\
&\geq& \phi^\tau(\bar t_{k},\bar x_k,\bar p_k)+\langle\frac{\partial \phi^\tau }{\partial p}(\bar t_{k},\bar x_k,\bar p_k),\tilde p-\bar p_k\rangle+2\delta\eta^2+\langle \hat p,p-\tilde p \rangle\\
&\geq& \phi^\tau(\bar t_{k},\bar x_k,\bar p_k)+\langle\frac{\partial \phi^\tau }{\partial p}(\bar t_{k},\bar x_k,\bar p_k), p-\bar p_k\rangle+2\delta\eta^2+\langle \hat p-\frac{\partial \phi^\tau }{\partial p}(\bar t_{k},\bar x_k,\bar p_k),p-\tilde p \rangle.
\end{eqnarray*}
Since $\frac{\partial \phi^\tau }{\partial p}(\bar t_{k},\bar x_k,\bar p_k)\in\partial{V^\tau}^-(\bar t_k,\bar x_k, \bar p_k)$ and $p-\tilde p=c (p-\bar p_k)$ $(c>0)$ and $V^\tau$ is convex in $p$ it holds
\[\langle \hat p-\frac{\partial \phi^\tau }{\partial p}(\bar t_{k},\bar x_k,\bar p_k),p-\tilde p \rangle\geq0.\]
So we have for all $p\in\Delta(I)\setminus \textnormal{Int}(B_\eta(\bar p_k))$
\begin{eqnarray}
V^\tau(\bar t_{k},\bar x_k,p) \geq \phi^\tau(\bar t_{k},\bar x_k,\bar p_k)+\langle\frac{\partial \phi^\tau }{\partial p}(\bar t_{k},\bar x_k,\bar p_k), p-\bar p_k\rangle+2\delta\eta^2
\end{eqnarray}
which gives in the limit for all $p\in\Delta(I)\setminus \textnormal{Int}(B_\eta(\bar p))$
\begin{eqnarray}
w(\bar t,\bar x,p) \geq \phi(\bar t,\bar x,\bar p)+\langle\frac{\partial \phi }{\partial p}(\bar t,\bar x,\bar p), p-\bar p\rangle+2\delta\eta^2.
\end{eqnarray}
Assume now that (28) does not hold for a $p\in\Delta(I)$.
Hence there exists a sequence $(\tau,x_{k_n},p_{k_n})\rightarrow(0, 0, p)$ with $\tau=\frac{T}{n}$, ${p_{k_n}}\in\Delta(I)\setminus B_\eta(\bar p_{k_n})$,  such that
\begin{eqnarray*}
 &&V^\tau(\bar t_{k_n}+\tau,\bar x_{k_n}+x_{k_n},p_{k_n})\\
 &&\ \ \ \ \ \ < \phi^\tau(\bar t_{k_n}+\tau,\bar x_{k_n}+x_{k_n},\bar{p}_{k_n})+\langle\frac{\partial \phi^\tau}{\partial p}(\bar t_{k_n}+\tau,\bar x_{k_n}+x_{k_n},p_{k_n}),p_{k_n}-\bar p_{k_n}\rangle+\delta|p_{k_n}-\bar p_{k_n}|^2
\end{eqnarray*}
Thus for $n\rightarrow\infty$, $p\in\Delta(I)\setminus \textnormal{Int}(B_\eta(\bar p))$ and
\begin{eqnarray}
w(\bar t,\bar x,p)< \phi(\bar t,\bar x,\bar p)+\langle\frac{\partial \phi }{\partial p}(\bar t,\bar x,\bar p), p-\bar p\rangle+\delta\eta^2
\end{eqnarray}
which contradicts (31).\\

In the following we denote
\begin{eqnarray*}
\bar X_{k+1}=\bar x_k+\sigma(\bar t_k,\bar x_k) \Delta B^{l_k}.
\end{eqnarray*}
where $\Delta B^{l_k}= B_{\bar t_k+\tau}-B_{\bar t_k}$.
With the estimate (28) we have for $\tau$ small enough for all $p\in\Delta(I)$
\begin{eqnarray*}
&&\mathbb E\left[V^\tau(\bar t_{k}+\tau,\bar X_{k+1},p)\right]\\
&&=\mathbb E\left[V^\tau(\bar t_{k}+\tau ,\bar X_{k+1},p)1_{|\bar X_{k+1}-\bar x_k|< \eta}\right]+\mathbb{E}\left[V^\tau(\bar t_{k}+\tau,\bar X_{k+1},p)1_{|\bar X_{k+1}-\bar x_k|\geq \eta}\right]\\
&&\geq\mathbb E\left[\left(\phi^\tau(\bar t_{k}+\tau,\bar X_{k+1},\bar p_k)+\langle \frac{\partial}{\partial p} \phi^\tau(\bar t_{k}+\tau,\bar X_{k+1},\bar p_k),p-\bar p_k\rangle+\delta |p-\bar p_k|^2\right)1_{|\bar X_{k+1}-\bar x_k|< \eta}\right]\\
&&\ \ \ \ \ \ +\mathbb{E}\left[\phi^\tau(\bar t_{k}+\tau,\bar X_{k+1},p)1_{|\bar X_{k+1}-\bar x_k|\geq \eta}\right]\\
&&=\mathbb E\left[\phi^\tau(\bar t_{k}+\tau,\bar X_{k+1},\bar p_k)+\langle \frac{\partial}{\partial p} \phi^\tau(\bar t_{k}+\tau,\bar X_{k+1},\bar p_k),p-\bar p_k\rangle+\delta 1_{|\bar X_{k+1}-\bar x_k|< \eta} |p-\bar p_k|^2)\right]\\
&&\ \ \ \ \ \ +\mathbb{E}\bigg[1_{|\bar X_{k+1}-\bar x_k|\geq\eta}\bigg(\phi^\tau(\bar t_{k}+\tau,\bar X_{k+1},p)-\phi^\tau(\bar t_{k}+\tau,\bar X_{k+1},\bar p_k) \\
&&\ \ \ \ \ \ \ \ \ \ \ \ \ \ \ \ \ \ \ \ \ \ \ \ \ \ \ \ \ \ \ \ \ \ \ \ \ \ \ \ \ \ \ \ \ \ \ -\langle \frac{\partial}{\partial p} \phi^\tau(\bar t_{k}+\tau,\bar X_{k+1},\bar p_k),p-\bar p_k\rangle\bigg)\bigg].
\end{eqnarray*}
Recalling that $\phi^\tau$ is convex with respect to $p$, we get for all $p\in\Delta(I)$
\begin{equation}
\begin{array}{rcl}
\mathbb E\left[V^\tau(\bar t_{k}+\tau,\bar X_{k+1},p)\right]&\geq&\mathbb E\bigg[\phi^\tau(\bar t_{k}+\tau,\bar X_{k+1},\bar p_k)+\langle \frac{\partial}{\partial p} \phi^\tau(\bar t_{k}+\tau,\bar X_{k+1},\bar p_k),p-\bar p_k\rangle\\
&&\ \ \ \ \ \ \ \ \ \ \ \ \ \ \ \ \ \ \ \ \ +\delta 1_{|\bar X_{k+1}-\bar x_k|< \eta} |p-\bar p_k|^2)\bigg].
\end{array}
\end{equation}

\textbf{Step 2}: 
Next we establish an estimate for  $\bold p_{k+1}:=\bold p_{l_k+1}^{\bar p_k,\bar x_k}$ where $\bold p_{l_k+1}^{\bar p_k,\bar x_k}$ is defined as one step martingale as in Definition 4.2. with initial data $(\bar t_k,\bar x_k,\bar p_k)$.\\
Note that by the one step dynamic programming (21) it holds
\begin{eqnarray}
V^\tau(\bar t_k,\bar x_{k}, \bar p_k)
=\mathbb{E}\left[ V^\tau(\bar t_{k}+\tau,\bar X_{k+1},\bold p_{k+1})+\tau H(\bar t_{k},\bar x_k,\bar z_{k}(\bar x_k,\bold p_{k+1}),\bold p_{k+1})\right].
\end{eqnarray}
Together with $V^\tau(\bar t_k,\bar x_{k}, \bar p_k)=\phi^\tau(\bar t_k,\bar x_{k}, \bar p_k)$ and the estimate (33) we have for all small enough $\tau>0$
\begin{eqnarray*}
\phi^\tau(\bar t_k,\bar x_{k}, \bar p_k)
&\geq&\mathbb{E}\Bigg[ \phi^\tau(\bar t_{k}+\tau,\bar X_{k+1},\bar p_{k})+\tau H(\bar t_{k},\bar x_k,\bar z_{k}(\bar x_k,\bold p_{k+1}),\bold p_{k+1})\\
&&\ \ \ \ \ \ \ +\langle \frac{\partial}{\partial p}  \phi^\tau(\bar t_{k}+\tau,\bar X_{k+1},\bar p_k),\bold p_{k+1}-\bar p_k\rangle+\delta  1_{|\bar X_{k+1}-\bar x_k|< \eta}  |\bar p_k-\bold p_{k+1}|^2\Bigg].
\end{eqnarray*}
Since $\bold p_{k+1}$ and $\Delta B^{l_k}$ are independent, $\phi^\tau$ has bounded derivatives and $\bold p_{k+1}$ is a one step martingale, it holds
\begin{eqnarray*}
&&\mathbb{E}\left[ \langle \frac{\partial}{\partial p}  \phi^\tau(\bar t_{k}+\tau,\bar X_{k+1},\bar p_k),\bold p_{k+1}-\bar p_k\rangle\right]\\
&&\ \ \ \ =\mathbb{E}\left[ \langle \frac{\partial}{\partial p}  \phi^\tau(\bar t_{k}+\tau,\bar x_k+\sigma(\bar t_k,\bar x_k) \Delta B^{l_k},\bar p_k),\bold p_{k+1}-\bar p_k\rangle\right]=0.
\end{eqnarray*}
Furthermore by the Markovian inequality and assumption (A) we have
\begin{eqnarray*}
&&\mathbb{E}\left[1_{|\bar X_{k+1}-\bar x_k|< \eta} |\bold p_{k+1}-\bar p_k|^2\right]\\
&&\ \ \ \ =\mathbb{E}\left[1_{|\sigma(\bar t_k,\bar x_k) \Delta B^{l_k}|< \eta}|\bold p_{k+1}-\bar p_k|^2\right]\geq c (1-\tau^\frac{1}{2}) \mathbb{E}\left[|\bold p_{k+1}-\bar p_k|^2\right].
\end{eqnarray*}
with a sufficiently small constant $c$ independent of $k$. Thus
\begin{equation}
\begin{array}{rcl}
0&\geq&\mathbb{E}\Bigg[ \phi^\tau(\bar t_{k}+\tau,\bar X_{k+1},\bar p_k)-\phi^\tau(\bar t_k,\bar x_{k}, \bar p_k)+\tau H(\bar t_{k},\bar x_k,\bar z_{k}(\bar x_k,\bold p_{k+1}),\bold p_{k+1})\\
&& \ \ \ \ \ \ \ \ +c \delta (1-\tau^\frac{1}{2})|\bold p_{k+1}-\bar p_k|^2\Bigg].
\end{array}
\end{equation}
Since $\phi^\tau$ has bounded derivatives it holds by assumption (A)
\begin{eqnarray}
\left|\mathbb{E}\left[\phi^\tau(\bar t_{k}+\tau,\bar X_{k+1},\bar p_k)-\phi^\tau(\bar t_k,\bar x_{k}, \bar p_k)\right]\right|&\leq& c\tau
\end{eqnarray}
and since $\mathbb{E}\left[|\bar z_{k}(\bar x_k,\bold p_{k+1})|\right]\leq c$ (16) it holds by (A) and  H\"older
\begin{equation}
\begin{array}{rcl}
\mathbb{E}\left[\tau H(\bar t_{k},\bar x_k,\bar z_{k}(\bar x_k,\bold p_{k+1}),\bold p_{k+1})\right]&\leq& c \tau.
\end{array}
\end{equation}
Combining (35)-(37) we have for small enough $\tau>0$  and a generic constant $c'>0$
\begin{eqnarray*}
\mathbb{E}[|\bold p_{k+1}-\bar p_k|^2]\leq \frac{c'}{c \delta(1-\tau^\frac{1}{2})}\tau,
\end{eqnarray*}
hence for small enough $\tau$ and a constant $c''>0$
\begin{eqnarray}
\mathbb{E}[|\bold p_{k+1}-\bar p_k|^2]\leq c''\tau.
\end{eqnarray}

\textbf{Step 3}: 

Furthermore we have with (35) and the monotonicity Lemma 3.3., since $V^\tau(\bar t_k,\bar x_{k}, \bar p_k)=\phi^\tau(\bar t_k,\bar x_{k}, \bar p_k)$
\begin{eqnarray}
0
&\geq&\mathbb{E}\bigg[ \phi^\tau(\bar t_{k}+\tau,\bar X_{k+1},\bold p_{k+1})-\phi^\tau(\bar t_k,\bar x_{k}, \bar p_k)+\tau H(\bar t_{k},\bar x_k,\tilde z_{k}(\bar x_k,\bold p_{k+1}),\bold p_{k+1})\bigg]-\tau \mathcal{O}(\tau),
\end{eqnarray}
where
\begin{eqnarray*}
\tilde z_{k}(\bar x_k,\bold p_{k+1})=\frac{1}{\tau}\mathbb{E}\left[\phi^{\tau}(\bar t_{k}+\tau,\bar X_{k+1},p) (\sigma^*)^{-1}(\bar t_k,\bar x_k)\Delta B^{l_k}\right]\big|_{p=\bold p_{k+1}}.
\end{eqnarray*}

From the construction of $\bold p_{k+1}$ and the fact that $\phi^\tau$  is convex it holds with (27)
\begin{equation}
\begin{array}{rcl}
\mathbb{E}\left[ \phi^\tau(\bar t_{k}+\tau,\bar X_{k+1},\bold p_{k+1})\right]
&\geq&\mathbb E\left[\phi^\tau(\bar t_{k}+\tau,\bar X_{k+1},\bar p_k)+\langle \frac{\partial}{\partial p} \phi^\tau(\bar t_{k}+\tau,\bar X_{k+1},\bar p_k),\bold p_{k+1}-\bar p_k\rangle\right]\\
\ \\
&=&\mathbb E\left[\phi^\tau(\bar t_{k}+\tau,\bar X_{k+1},\bar p_k)\right].
\end{array}
\end{equation}

It remains to get a suitable estimate for $\tilde z_{k}(\bar x_k,\bold p_{k+1})$. Since $\phi^{\tau}$ is uniformly Lipschitz continuous in $x$, it holds by Taylor expansion in $x$
\begin{eqnarray*}
\tilde z_{k}(\bar x_k,\bold p_{k+1})&=&\frac{1}{\tau}\mathbb{E}\left[\phi^{\tau}(\bar t_{k}+\tau,\bar X_{{k+1}}, p) (\sigma^*)^{-1}(\bar t_k,\bar x_k)\Delta B^{l_k}\right]\big|_{p=\bold p_{k+1}}\\
&=&\frac{1}{\tau}\mathbb{E}\left[\phi^{\tau}(\bar t_{k}+\tau,\bar x_k,p) (\sigma^*)^{-1}(\bar t_k,\bar x_k)\Delta B^{l_k}\right]\big|_{p=\bold p_{k+1}}\\
&&\ \ \ \ \  +\frac{1}{\tau}\mathbb{E}\left[D_x \phi^{\tau}(\bar t_{k}+\tau,\bar x_k,p)|\Delta B^{l_k}|^2\right]\big|_{p=\bold p_{k+1}}+ \mathcal{O}(\tau)\\
&=&\frac{1}{\tau}\mathbb{E}\left[D_x \phi^{\tau}(\bar t_{k}+\tau,\bar x_k,p)|\Delta B^{l_k}|^2\right]\big|_{p=\bold p_{k+1}}+ \mathcal{O}(\tau).
\end{eqnarray*}
Furthermore since $D_x\phi^{\tau}$ is Lipschitz continuous in $p$ it holds with (38)
\begin{eqnarray*}
&&\mathbb{E}\left[\left|D_x \phi^{\tau}(\bar t_{k}+\tau,\bar x_k,\bold p_{k+1})|\Delta B^{l_k}|^2-D_x \phi^{\tau}(\bar t_{k}+\tau,\bar x_k,\bar p_{k})|\Delta B^{l_k}|^2\right|\right]\\
&&\ \ \ \ \ \ \ \ \ \ \ \ \leq c\mathbb{E}\left[|\bold p_{k+1}-p_k||\Delta B^{l_k}|^2\right]\leq c\tau^\frac{3}{2}
\end{eqnarray*}
So from (40) we have
\begin{eqnarray*}
0&\geq& \mathbb{E}\Bigg[ \phi^\tau(\bar t_{k}+\tau,\bar X_{k+1},\bar p_k)-\phi^\tau(\bar t_k,\bar x_{k}, \bar p_k)+\tau H(\bar t_{k},\bar x_k, D_x \phi^{\tau}(\bar t_{k}+\tau,\bar x_k,\bar p_{k}) ,\bar p_{k})\Bigg]\\
&&\ \ \ \ -c({\tau}^\frac{3}{2}+\tau \mathcal{O}(\tau))
\end{eqnarray*}
which implies (25) since $\phi^\tau$ is equal to $\phi$ up to a linear shift.

\end{proof}

\section{Concluding Remarks and Outlook}

In this paper we gave an approximation scheme for the value function of a stochastic differential game with incomplete information. It is natural to ask whether this approximation might be used to determine optimal feedback strategies for the informed player. In the deterministic games with complete information it is well known that the answer is positive (see the step by step motions associated with feedbacks in \cite{K}). The case of deterministic games with incomplete information has been treated in \cite{Carda}.\\
The approximation of optimal strategies for stochastic differential games is a more delicate topic even in the case with complete information.  \cite{Ba} - also considering the game under a Girsanov transform -  gives a partwise answer under a weak Lipschitz assumption of the feedback control. The result is shown by using approximations of BSDEs however not in a completely discrete framework. In the very recent paper \cite{FH} approximately Markov strategies are constructed with an approximation that in contrast to ours takes into account the actions of the other player during the time intervals. This however makes the approximation much harder to implement.\\
In fact, if we use the approximation for the construction of optimal strategies for the informed player we are in the same situation as \cite{Ku}. For the approximation of the value function in \cite{Ku} nearly optimal policies are constructed which possess a certain optimality in the approximative discrete time games instead of the continuous time one. To the authors knowledge the problem of finding an efficient approximation of optimal strategies in stochastic differential games (with or without incomplete information) is open and poses an interesting problem for further research.


\end{document}